\newif\ifarxiv
\newif\ifpodc
\arxivtrue 
\podcfalse 

\ifarxiv
\documentclass[11pt]{article}
\usepackage[letterpaper, margin=1in]{geometry}

\usepackage[english]{babel}
\usepackage[T1]{fontenc} 

\usepackage{xspace}
\usepackage{algorithm}
\usepackage{algpseudocode}
\usepackage{amsmath}
\usepackage{graphicx}
\usepackage[colorlinks=true, allcolors=blue]{hyperref}

\usepackage[utf8]{inputenc} 
\usepackage{url}            
\usepackage{booktabs}       
\usepackage{amsfonts}       
\usepackage{nicefrac}       
\usepackage{microtype}      
\usepackage[dvipsnames]{xcolor}         
\usepackage{graphicx}

\usepackage{amsthm}
\usepackage{amsmath,amssymb}
\usepackage{comment}        
\usepackage{enumitem}  
\usepackage{cite} 
\usepackage{parskip}

\usepackage{natbib}
\newtheorem{definition}{Definition}[section]
\newtheorem{claim}{Claim}
\newtheorem{theorem}{Theorem}
\newtheorem*{theorem*}{Theorem}
\newtheorem{proposition}[theorem]{Proposition}

\newtheorem{lemma}[theorem]{Lemma}
\newtheorem*{lemma*}{Lemma}

\newtheorem{remark}{Remark}[section]

\algdef{SE}[DOWHILE]{Do}{doWhile}{\algorithmicdo}[1]{\algorithmicwhile\ #1}%
\title{Breadth-First Depth-Next:  \protect\\ Optimal Collaborative Exploration  \protect\\of Trees with Low Diameter}
\fi

\ifpodc 
\documentclass[acmsmall,nonacm,anonymous]{acmart} 
\usepackage{algorithm}
\usepackage{algpseudocode}
\algdef{SE}[DOWHILE]{Do}{doWhile}{\algorithmicdo}[1]{\algorithmicwhile\ #1}%

\newtheorem{claim}{Claim} 
\newtheorem*{lemma*}{Lemma} 

\AtEndPreamble{%
\theoremstyle{acmplain}
\newtheorem{remark}[theorem]{Remark}}

\title{Breadth-First Depth-Next: Optimal Collaborative Exploration of Trees with Low Diameter}
\fi

\newcommand{\Acal}{\mathcal{A}}
\newcommand{\Bcal}{\mathcal{B}}

\newcommand{\Dcal}{\mathcal{D}}

\newcommand{\Ocal}{\mathcal{O}}

\newcommand{\Tcal}{\mathcal{T}}

\newcommand{\Nbb}{\mathbb{N}}

\newcommand{\T}{\mathsf{T}}

\usepackage{bbm}

\newcommand{\minkdel}{\min\{\log \Delta, \log k\}}
\newcommand{\depth}{\delta}

\newcommand{\card}[1]{\left|{#1}\right|}
\newcommand{\set}[1]{\left\{{#1}\right\}}
\newcommand{\ceil}[1]{{\lceil #1 \rceil}}
\newcommand{\floor}[1]{{\lfloor #1 \rfloor}}
\newcommand{\lca}{\text{LCA}}

\DeclareMathOperator{\DFBNmath}{\texttt{BFDN}}

\algblock{ParFor}{EndParFor}
\algrenewtext{ParFor}[1]{\textbf{for} #1 \textbf{do in parallel}}
\algrenewtext{EndParFor}{\textbf{end\ for}}

\ifpodc


\fi

\begin{document}

\ifarxiv
\author{
  Romain Cosson\\
  \small{\texttt{romain.cosson@inria.fr}}
  \and
  Laurent Massoulié\\
  \small{\texttt{laurent.massoulie@inria.fr}}
  \and
  Laurent Viennot\\
  \small{\texttt{laurent.viennot@inria.fr}}
}
\date{}
\maketitle
\fi

\begin{abstract}
We consider the problem of \textit{collaborative tree exploration} posed by Fraigniaud, Gasieniec, Kowalski, and Pelc~\citep{fraigniaud2006collective} where a team of $k$ agents is tasked to collectively go through all the edges of an unknown  tree as fast as possible.
Denoting by $n$ the total number of nodes and by $D$ the tree depth, the $\mathcal{O}(n/\log(k)+D)$ algorithm of \cite{fraigniaud2006collective} achieves the best-known competitive ratio with respect to the cost of offline exploration which is $\Theta(\max{\{2n/k,2D\}})$. Brass, Cabrera-Mora, Gasparri, and Xiao \cite{brass2011multirobot} consider an alternative performance criterion, namely the additive overhead with respect to $2n/k$, and obtain a $2n/k+\mathcal{O}((D+k)^k)$ runtime guarantee.
In this paper, we introduce `Breadth-First Depth-Next' (BFDN), a novel and simple algorithm that performs collaborative tree exploration in time $2n/k+\mathcal{O}(D^2\log(k))$, 
thus outperforming \cite{brass2011multirobot} for all values of $(n,D)$ and being order-optimal for all trees with depth $D=o_k(\sqrt{n})$.
Moreover, a recent result from 
\cite{disser2017general} implies that no exploration algorithm can achieve a $2n/k+\mathcal{O}(D^{2-\epsilon})$ runtime guarantee. The dependency in $D^2$ of our bound is in this sense optimal.
The proof of our result crucially relies on the analysis of an associated two-player game. We extend the guarantees of BFDN to: scenarios with limited memory and communication, adversarial setups where robots can be blocked, and exploration of classes of non-tree graphs. Finally, we provide a recursive version of BFDN with a runtime of $\mathcal{O}_\ell(n/k^{1/\ell}+\log(k) D^{1+1/\ell})$ for parameter $\ell\ge 1$, thereby improving performance for trees with large depth. 
\end{abstract}

\ifpodc
\maketitle
\fi

\section{Introduction}
\paragraph{Problem setting.}
We consider the problem of \textit{collaborative tree exploration} posed by \citep{fraigniaud2006collective} where a team of agents, or robots\footnote{these terms will be used indistinctly, but the term ``robots'' is often preferred in line with the initial work of \cite{fraigniaud2006collective}.} is tasked to collectively go through all the edges of a tree as fast as possible and then return to the root. At initialization, all robots are located at the root. At each round, every robot can move along one incident edge to reach a neighbour. Edges are revealed along with their unique identifier when a robot reaches one of their endpoints. 
We assume that robots can communicate and compute at no cost. In particular, at every instant, they share a common map of the sub-tree that has already been explored. 
Collaborative online tree exploration is a fundamental problem behind the exploration of an unknown physical environment with a team of agents. It has been intensively studied for the special case of trees~\cite{fraigniaud2006collective,higashikawa2014online,brass2011multirobot,dynia2006smart,dynia2007robots, dereniowski2015fast, OrtolfS14}. Despite almost two decades of research, we still do not completely understand how much exploration time can decrease with the number of agents.

\paragraph{Main results.} In this paper, we present a simple and novel algorithm that achieves collaborative tree exploration with $k$ agents in $\frac{2n}{k}+D^2(\min\{\log(k),\log(\Delta)\}+2)$ rounds for any tree with $n$ nodes, depth $D$ and maximum degree $\Delta$. 

The algorithm is called ``Breadth-First Depth-Next'' (abbreviated \texttt{BFDN}) and the behaviour of the robots can synthetically be described as follows: when located at the root, a robot is sent to a node which is adjacent to the highest unexplored edge (breadth-first mode). Upon arrival, the robot changes behaviour: it will choose to go through an unexplored edge if one is adjacent to its position and will go one step back towards the root otherwise (depth-next mode). 

The analysis of this algorithm involves a simple, yet non-trivial, zero-sum two-player game opposing a player and an adversary moving $k$ balls in $k$ urns. We show that a simple strategy for this game induces a cost of at most $k\min\{\log(k),\log(\Delta)\}+k$ to the player.

Our algorithm is easy to implement and can also be adapted to more complex settings, such as i) exploration of specific classes of non-tree graphs, ii) scenarios with constrained communications and memory including the classical local communication model, or iii) setups where an adversary chooses at each time step which robots are allowed to move. 

We provide simple extensions of \texttt{BFDN} for all three settings in Section~\ref{sec: other-settings}. Finally, in an attempt to improve dependence in $D$, we propose \texttt{BFDN}$_\ell$, a recursive version of \texttt{BFDN} that explores the tree in time $\Ocal_\ell(\frac{n}{k^{1/\ell}}+\min\{\log(k),\log(\Delta)\}D^{1+1/\ell})$ where $\ell\geq 1$ is some constant provided as input.

\paragraph{Useful context and related works.} In the case of a single robot, exploration can be performed optimally using the classical ``Depth First Search'' (\texttt{DFS}) strategy. This strategy can be implemented by having the robot always  go through an adjacent unexplored edge if one is available and go one step up towards the root otherwise. This strategy has the robot go through all edges of the tree exactly twice and return to the origin in a total of exactly $2(n-1)$ rounds, where $n$ is the number of nodes. Note that it can be implemented both \textit{offline} (if the tree is known in advance) and \textit{online} (if edges are revealed when their endpoint is reached). 

In the multi-robot setting, with $k\geq 2$, the best \textit{offline} exploration strategy takes at least $2(n-1)/k$ time-steps because all edges must be traversed at least twice by some robot (remember that robots must finish exploration at the root). Since nodes at depth $D$ must also be attained, we can lower-bound the time complexity of offline collaborative tree exploration by $\max \{2n/k,2D\} \geq n/k + D$. A simple algorithm by \cite{dynia2006power, OrtolfS14} matches this bound up to a factor $2$: consider a depth-first search path from the root of length $2(n-1)$, and divide it into $k$ subsets each of length $2(n-1)/k$ then assign each robot to go through one of these segments and go back to the origin in total time $2(n/k + D)$. The complexity of offline collaborative exploration is thus $\Theta(n/k + D)$. 
Interestingly, the optimal runtime is NP-hard to compute as \cite{fraigniaud2006collective} gave a reduction from this problem to  \texttt{3-PARTITION}.

To analyze the \textit{online} multi-robot exploration problem, the literature initially focused on the \textit{competitive ratio} which is the worst-case ratio between the cost of the online and the optimal offline algorithm. For an online algorithm $\mathcal{A}_k$ using $k\geq 2$ robots, this ratio is defined up to a constant factor as $\max_{n,D \in \Nbb}\max_{T\in \Tcal(n,D)}\texttt{Runtime}(T,\mathcal{A}_k)/(n/k + D)$ where $\Tcal(n,D)$ denotes the set of all trees with $n$ nodes and depth $D$. The algorithm proposed initially by \cite{fraigniaud2006collective} \texttt{CTE} (Collective Tree Exploration) has a runtime of $O(\frac{n}{\log k}+D)$ and therefore attains a competitive ratio of $O(\frac{k}{\log k})$. It was later shown by \cite{higashikawa2014online}  that the competitive analysis of \texttt{CTE} was tight as they provided a simple construction of a tree with $n = kD$ edges that \texttt{CTE} would take $\frac{Dk}{\log_2(k)}$ time-steps to explore.  To date, no algorithm is known to have a better competitive ratio than \texttt{CTE}, while the best known lower-bound on the quantity, for deterministic exploration algorithms, is in $\Omega(\frac{\log k}{\log \log k})$ by \cite{dynia2007robots}. 

Given the difficulty of characterizing the unconditional worst-case competitive ratio,  a line of work has focused on obtaining upper and lower bounds on the competitive ratio conditional on values of $n, D$ \cite{OrtolfS14,brass2011multirobot, dynia2006smart, disser2017general, dereniowski2015fast, higashikawa2014online}. Their goal is thus to minimize the total runtime as a function of $(n,D)$. In this spirit, \cite{brass2011multirobot} proposed a novel analysis of \texttt{CTE} yielding a guarantee in $\frac{2n}{k}+\Ocal((k+D)^k)$, hence an optimal dependence in $(n,k)$ with a large additive cost which does not depend on $n$. On the other hand, \cite{OrtolfS14} derived a recursive algorithm called \texttt{Yo}* that runs in time $\Ocal(2^{\Ocal(\sqrt{\log D \log \log k})}\log(k)(\log(k)+\log(n))(n/k+D))$, with optimal runtime up to sub-linear multiplicative factors. 

The algorithm we propose with its guarantee of $\frac{2n}{k}+ \Ocal(D^2\log(k))$  complements this line of work. Our guarantee yields a strict improvement over \cite{brass2011multirobot} for all values of $(n,D,k)$, and improves upon \texttt{CTE} and $\texttt{Yo}^*$ in specific ranges of parameters as depicted in Figure~\ref{fig:CTEvsBMRDFS}. 
\begin{figure}[h!]
  \centerline{\includegraphics[width=0.65\textwidth,trim={1cm 1cm 1cm 1cm}]{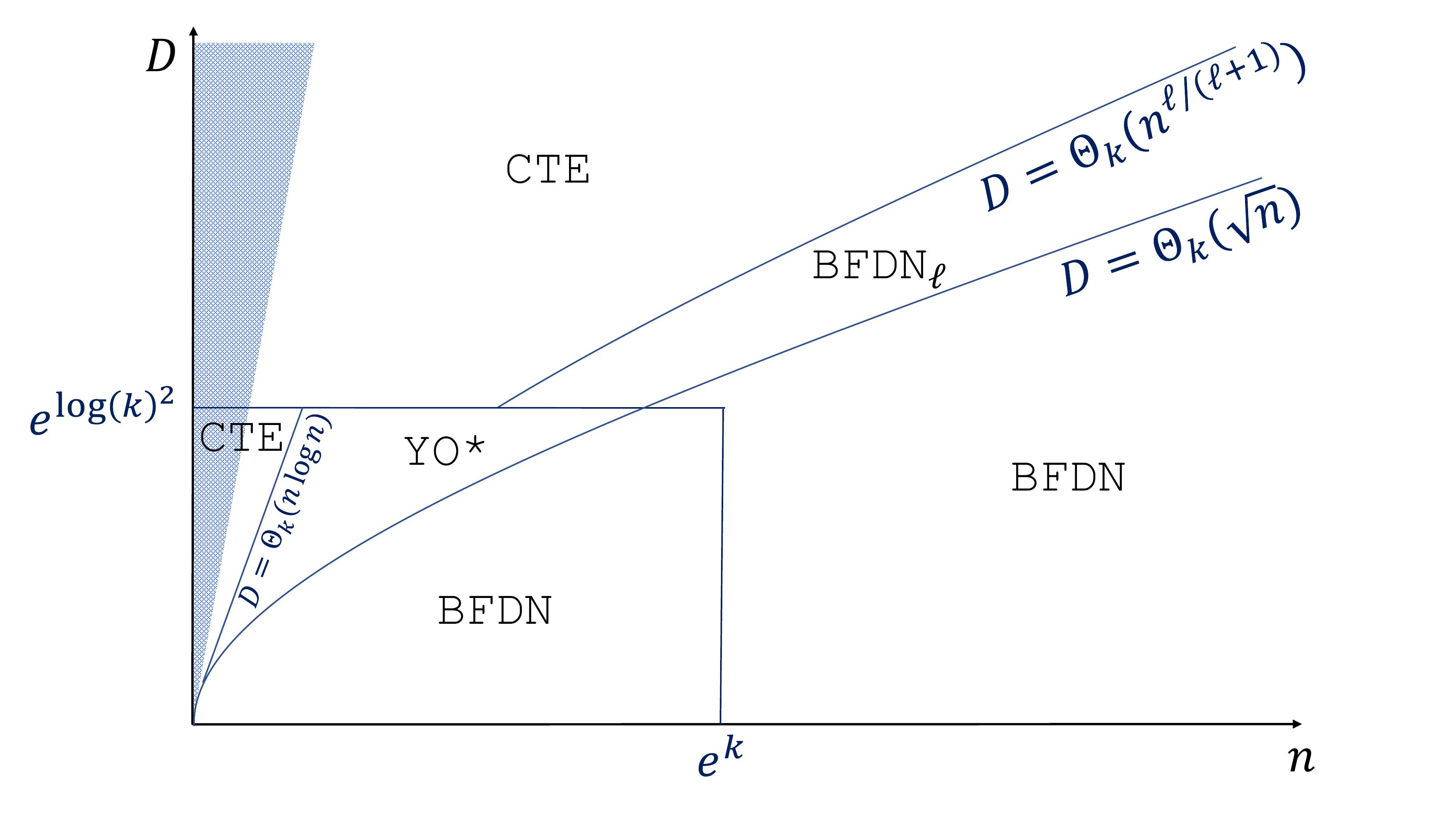}}
  \caption{Regions of parameters $(n,D)$ where either of \texttt{CTE}, $\texttt{Yo}^*$, \texttt{BFDN} and \texttt{BFDN}$_\ell$ has the best runtime guarantee, for fixed constant $k$ and $\ell =\Ocal(\log k/\log \log k)$. The runtime of algorithm Yo* was simplified to improve readability; see Appendix~\ref{sec: conclusion} for details. No trees can be defined in the dashed region where $n\leq D$.}
  \label{fig:CTEvsBMRDFS}
\end{figure}

In line with \cite{brass2011multirobot}, our work provides further motivation to study the \textit{additive overhead} of collaborative \textit{online} exploration rather than the \textit{multiplicative overhead} (competitive ratio). In a recent work,~\cite{disser2017general} showed that collaborative exploration requires at least time $\Omega(D^2)$ for a specific adversarial construction of trees and for $k=n$. This implies that for arbitrary $\epsilon>0$,  no deterministic collaborative exploration algorithm can achieve runtime of $\frac{2n}{k}+\Ocal(D^{2-\epsilon})$. In this sense, the dependency in $D^2$ of our guarantee $\frac{2n}{k}+\Ocal(D^2\log(k))$ is close-to-optimal. 

Collaborative tree exploration has also been studied under additional assumptions. For sparse trees, such as trees that can be embedded in the $2$-dimensional grid, \cite{dynia2006smart} obtained a runtime of $\Ocal(\sqrt{D}(\frac{n}{k}+D))$. More generally, their bound leads to a competitive ratio of $\Ocal(D^{1-1/p}\cdot\min\{p, \log p D^{1/2p}\})$ where $p$ is the {\em sparsity parameter} of the underlying tree. Also, \cite{dereniowski2015fast} investigated the collaborative tree exploration problem under the assumption that the number of robots $k$ is very large instead of being a fixed constant, specifically $k\geq Dn^c$ for some constant $c>1$. In this setting, and assuming global communication, their algorithm achieves exploration in $\frac{c}{c-1}D+o(D)$. Interestingly, their guarantees also apply to the more challenging and less studied collaborative graph exploration problem; see also \cite{brass2011multirobot,brass2014improved}.

Finally we note that collaborative exploration is loosely connected to the notion of $k$-cover time of random walks. The cover time of a random walk on a graph is defined as the time until all nodes have been visited at least once, see \cite{aldous1995reversible} chapter 6 for an analysis of the cover time by multiple random walkers, extending a work of \cite{broder1989trading}. In the same spirit, \cite{alon2008many} showed that the $k$-cover time of independent random walks scales in $1/k$ for several families of graphs, including expanders and Erd\H{o}s-Rényi random graphs. However these  results of \cite{alon2008many} do not apply to trees. Moreover the cover time of a tree by a random walker can be as large as $\Omega(n^2)$. 

\paragraph{Structure of paper.} 
Section \ref{sec:BFDN} describes the \texttt{BFDN} algorithm and provides the main results. Section \ref{sec: game} describes and analyzes a 2 player zero-sum board game, an essential ingredient in our proof of runtime bounds for \texttt{BFDN}. Section \ref{sec: other-settings} contains extensions of  \texttt{BFDN} to settings with: limited communications; adversarial interruption of robots; and non-tree graph exploration.  Finally, Section \ref{sec:trading} provides a recursive version  of  \texttt{BFDN} that yields improved runtime guarantees when $D$ gets larger with compared to $n$.

\paragraph{Notations.} The notation $\log(\cdot)$ refers to the natural logarithm and $\log_2(\cdot)$ to the logarithm in base $2$. For two integers $k$ and $D$ we will use the abbreviations $[k] = \{1,\dots,k\}$ and $[D[ =\{0,\dots,D-1\}$.

A tree $T=(V,E)$ is defined by its set of nodes $V$ and edges $E \subset V\times V$; it is rooted at some specific node denoted $\texttt{root}\in V$ from which all robots start the exploration. For a node $v\in V$, $\depth(v)$ is the distance of $v$ to the \texttt{root} and $T(v)$ denotes the sub-tree of $T$ rooted at $v$ containing all the descendants of $v$.  
The depth of $T$ is $D=\max_{v\in V}\depth(v)$.
We will also use a notion of \textit{partially explored tree} (defined in Section \ref{sec:BFDN}) that will naturally inherit from the same definitions.

\section{The Breadth-First Depth-Next algorithm}\label{sec:BFDN}
Our main result on  \texttt{BFDN}, that will be described shortly, is the following
\begin{theorem}
\label{th:BFDN_main}
{\normalfont \texttt{BFDN}} achieves online exploration of any tree with $k$ robots in at most
$$\frac{2n}{k} + D^2(\min\{\log(\Delta),\log(k)\}+2)$$
rounds, where $\Delta$ is the maximum degree of the tree, $n$ is the number of nodes, and $D$ is the depth. \end{theorem}

Before describing the algorithm, we precise the setting and some definitions. Following \cite{fraigniaud2006collective}, we consider a synchronous model with all-to-all communications. A team of $k$ robots operates in rounds, and the runtime of an exploration algorithm is measured by the number of rounds executed before termination. Communications and memory constraints are not considered in this section, but are introduced along with additional formalism in Section \ref{sec:commmunication}.

\paragraph{Partially explored tree.} 
At a given exploration round, $V$ denotes the set of \textit{discovered nodes}, i.e.   
nodes that have been occupied by at least one robot in the past, and  $E$ denotes the set of \textit{discovered edges}, i.e. edges that have a discovered endpoint. An \textit{unexplored edge} or \textit{dangling edge} is a discovered edge that has never been traversed by any robot. It can be viewed as a pair $(u,?)$, with $u\in V$. 
The \textit{partially explored tree} $T_{\text{online}}=(V,E)$ thus encodes all the information gathered by the robots. If there are no more dangling edges in $T_{\text{online}}$, it implies that exploration is complete and that the partially explored tree equals the underlying tree  $T_{\text{offline}}\in \Tcal(n,D)$.

\paragraph{Collaborative exploration algorithm.} A collaborative exploration algorithm under the full communication model is formally defined as a function that maps a partially explored tree $T=(V,E)$ as well as the list of positions of the agents $p_1,\dots,p_k \in V^k$ to a list of \textit{selected edges} $e_1,\dots,e_k \in (E\cup\{\perp\})^k$ that the agents will use for their next move. All selected edges $e_i\in E$ must be adjacent to the position $p_i$. Dangling edges may be selected. By convention, $\perp$ indicates that the corresponding agent will not move at the next round. In our pseudo-code, the routine $\texttt{SELECT}(\texttt{Robot}_i,e)$ performs the assignment $e_i\gets e$.  

When all agents have selected a next move, the routine \texttt{MOVE} is applied and all agents move along their selected edge synchronously. The partially explored tree $(V,E)$ is then updated with the new information provided by the agents that have traversed a dangling edge. 

For any rooted tree, exploration starts with all agents located at the root and the collaborative exploration algorithm is applied iteratively. The algorithm terminates when the explored tree $(V,E)$ contains no dangling edges and when the position of all agents is back at the root. The runtime of an exploration algorithm is defined as a function of $(n,D)$ by the number of rounds required before termination on any tree with $n$ nodes and depth $D$.

\paragraph{Breadth-First Depth-Next Algorithm.} 
\texttt{BFDN} is formally defined in Algorithm~\ref{alg:BFDN_algo}. Its description in words is as follows. When located at the root, a robot indexed by $i\in [k]$ and denoted $\texttt{Robot}_i$ is assigned an {\em anchor} $v_i\in V$ which is a discovered node that is adjacent to at least one dangling edge. If no such node exists, the anchor is the root itself.
The exact assignment process is specified by procedure \texttt{Reanchor} which gives the priority to nodes that are the closest to the root and that have the least number of anchored robots. $\texttt{Robot}_i$ then attains this anchor in a series of breadth-first moves performed with procedure \texttt{BF}. When the anchor is reached, the robot only makes depth-next moves until it returns to the root with procedure \texttt{DN}. In a sequence of depth-next moves, the robot always goes through a dangling edge if one is available (i.e. adjacent and not already selected as next move by another robot), and 
 goes one step up towards the root otherwise. This will result in a depth-first-like exploration inside $T(v_i)$. The algorithm stops when all robots are at the root and cannot be reassigned a new anchor because there are no more dangling edges. 
 
\begin{algorithm}
\caption{\texttt{BFDN} ``Breadth-First Depth-Next'' }\label{alg:BFDN_algo}
\begin{algorithmic}[1]
\Require $k$ robots that start at the root of an unknown  tree.
\Ensure The robots visit all nodes and return to the root.
\State $V  = $ list of discovered nodes ; $E = $ list of discovered edges ; $\depth(v) = $ depth of node $v\in V$
\State $v_i \gets \texttt{root} ~~ \forall i\in \{1,\dots,k\}$ \Comment{Initialize anchors.}
\State $S_i \gets [~] ~~ \forall i\in \{1,\dots,k\}$\Comment{Initialize empty stacks.}
\Do
\For{$i=1$ to $k$} \Comment{$\texttt{Robot}_i$ will select a move.}\label{line:for}
\If{$\texttt{Robot}_i$ is at \texttt{root}}
\State $v_i\gets \texttt{Reanchor}(i)$
\State Stack in $S_i$ the list of edges that lead to $v_i$
\EndIf
\If{$S_i$ is not empty}
\State $\texttt{BF}(i)$
\Else
\State $\texttt{DN}(i)$
\EndIf
\EndFor
\State \texttt{MOVE} all robots in their selected edge in parallel and update $(V,E)$ 
\doWhile{some robot changes position}\\

\Procedure{BF}{$i$}
\State Unstack $e\in E$ from $S_i$ and $\texttt{SELECT}(\texttt{Robot}_i,e)$
\EndProcedure\\

\Procedure{DN}{$i$}
\If{$\texttt{Robot}_i$ is adjacent to some dangling and unselected edge $e\in E$}
\State $\texttt{SELECT}(\texttt{Robot}_i,e)$
\Else
\State $\texttt{SELECT}(\texttt{Robot}_i,\texttt{up})$ \Comment{If $\texttt{Robot}_i$ is at the root, $\texttt{up}$ is interpreted as $\perp$.}
\EndIf 
\EndProcedure\\

\Procedure{Reanchor}{$i$}
\State $A = \{v \in V ~~ s.t. ~~ v \text{ is adjacent to some unexplored edge and } \delta(v) \text{ is minimal}\}$\label{line:anchorset}
\If{$A\neq \emptyset$}
\State $v_i\gets \arg\min_{v\in A} \#\{j ~~ s.t. ~~ v_j=v\}$ \Comment{Assigns to anchor with minimum number of robots.}
\Else
\State $v_i\gets \texttt{root}$ \Comment{The tree is explored, $\texttt{Robot}_i$ stays at the root.}
\EndIf
\EndProcedure
\end{algorithmic}
\end{algorithm}

\subsection{Analysis of \texttt{BFDN} and proof of Theorem~\ref{th:BFDN_main}}
We first prove the correctness and termination of \texttt{BFDN} and then bound its runtime. 

\paragraph{Correctness.} 

In Algorithm \ref{alg:BFDN_algo}, the main do-while loop is interrupted when no robot changes position at some round. Note that the \texttt{root} is the only place where robots may stay at the same position because direction $\texttt{up}$ is interpreted as $\perp$ at the root only. Thus all robots are at the root when the algorithm stops. Also note that the selection of direction \texttt{up} by all robots at the root implies that there are no dangling edges in the tree. Thus the tree has been entirely explored and all robots have returned to the origin. Thus, the algorithm is correct. 

\paragraph{Termination.} To prove termination, we show that while the algorithm runs, a node is discovered every $3D$ rounds at least. Since there are $n$ nodes in the tree, the algorithm must stop after at most $3D\times n$ rounds. Assume by contradiction that no node is discovered in a sequence of $3D$ rounds. After $2D$ rounds, all robots have attained the root because 
all \texttt{DF} moves are directed \texttt{up}. Thus, either one robot is assigned an anchor that is adjacent to an unexplored edge which will be traversed in the coming $D$ rounds, or the algorithm stops. In both cases we have a contradiction.

The time complexity analysis of \texttt{BFDN} the relies crucially on the following Lemma, that is
proved in Section~\ref{sec: game} using a reduction to a balls in urns game. 

\begin{lemma}\label{lem:main_lemma}
    In an execution of {\normalfont \texttt{BFDN}}, for any $d\in \{1,\dots,D-1\}$, the total number of reassignments over all $i\in[k]$ of  some ${\normalfont \texttt{Robot}}_i$'s anchor $v_i$ to a new value $v\in V$ of depth $\depth(v)=d$ is at most $k(\min\{\log(k),\log(\Delta)\}+2)$. 
\end{lemma}

\paragraph{Time complexity.} During the execution, a given $\texttt{Robot}_i$ anchored at $v_i$ can spend time in two different ways (1) not moving (2) moving along a selected edge. We denote $\T_i^1, \T_i^2$ the time (number of rounds) spent by $\texttt{Robot}_i$ in each of these phases. We have that $\sum_{i\in [k]}(\T_i^1+ \T_i^2) = k\T$ where $\T$ is the total number of rounds of the algorithm as the $k$ robots operate in parallel. We now  prove a series of claims.

\begin{claim}The total number of rounds when some robot does not move is at most $D+1$.
\end{claim}
\begin{proof}[Proof of claim 1]
    First note that if a robot does not move, it must be located at the root and have selected direction $\texttt{up}$ with procedure \texttt{DN}. This can only happen if the robot is also anchored at the root. Consequently, when a robot stays at the root, it means either that there are no more dangling edges in the explored tree (this happens at most $D$ times because all robots are on their way back) or that there are still dangling edges that are adjacent to the root, but they are all selected (this happens at most once because at the next time-step, all edges adjacent to the root will be explored). The number of time-steps when a robot may not move is thus at most $D+1$.
\end{proof}

\begin{claim}
When a dangling edge is explored for the first time, it is traversed by a single robot.
\end{claim}

\begin{proof}[Proof of claim 2:]
    All breadth-first moves (with procedure \texttt{BF}) are through previously explored edges because they lead from the root to a previously explored node. Thus dangling edges are only explored in depth-next moves (with procedure \texttt{DN}). In this procedure, a dangling edge is selected by a robot as its next move only if it is not selected as next move by some other robot. 
\end{proof} 

\begin{claim}
    Consider a sequence of moves by some $\texttt{Robot}_i$ that starts at the root with the assignment of an anchor $v$ of depth $\depth(v)=d \geq 0$ and that ends with the return of that robot to the root. We denote this sequence of moves $x$ and we denote by $\T_x$ its duration (in number of rounds). In such a sequence, $\texttt{Robot}_i$ has explored exactly $( \T_x-2d)/2$ dangling edges.
\end{claim}

\begin{proof}[Proof of claim 3]
The sequence of moves $x$ has the following structure. First, $\texttt{Robot}_i$ uses a shortest path from the \texttt{root} to $v$ which takes $d$ moves through previously explored edges. Then the robot performs moves inside $T(v)$ by going down through dangling edges if some are available and going up towards the root otherwise. Note that exactly half of the moves inside $T(v)$ must be through dangling edges as there must be as many moves down and up in $T(v)$. Finally, the robot goes back from $v$ to the root in again $d$ moves through explored edges. In the end, the robot has explored exactly $(\T_x-2d)/2$ dangling edges in this sequence.
\end{proof}

We now assemble the claims and Lemma \ref{lem:main_lemma} together to bound the time complexity of \texttt{BFDN}. Using claim 1, we have that $\sum_i \T_i^1 \leq k(D+1)$. Then, we write $\sum_i \T_i^2 = \sum_{d\in [D[}\sum_{x\in X_d}\T_x$ where $X_d$ represent the list of all sequences of moves $x$ that start with the assignment of some robot an anchor $v$ at depth $\depth(v)=d$ and that end with the return of that robot to the root. Using claim 2 and claim 3, we have that $\sum_{d\in [D[}\sum_{x\in X_d}(\T_x-2d)/2 \leq n-1$. Consequently, 
\begin{equation*}
\sum_{i\in [k]} \T_i^2 \leq 2(n-1) + 2\sum_{d\in [D[}\sum_{x\in X_d}d.
\end{equation*} 
By Lemma \ref{lem:main_lemma}, the cardinality of $X_d$ is at most $k(\min\{\log(k),\log(\Delta)\}+2)$, for $d\in \{1,\dots,D-1\}$. Thus, $\sum_{d\in [D[}\sum_{x\in X_d}d \leq \frac{D(D-1)}{2}k(\min\{\log(k),\log(\Delta)\}+2)$. Finally, using $\sum_{i\in [k]}(\T_i^1+ \T_i^2) = k\T$, we obtain $k\T\leq 2(n-1) + D(D-1)k(\min\{\log(\Delta),\log(k)\}+2)+(D+1)k$, which proves that the algorithm stops after at most
\begin{equation*}
    \T\leq \frac{2n}{k}+ D^2(\min\{\log(\Delta),\log(k)\}+2)
\end{equation*}
steps, thus completing Theorem~\ref{th:BFDN_main}'s proof. \\

Though it is not required for the the analysis above, we conclude this Section with a final claim that helps the general understanding of the algorithm.

\begin{claim}At all rounds, all dangling edges are in $\cup_{i\in [k]} T(v_i)$.
\end{claim}

\begin{proof}[Proof of claim 4]
Consider some dangling edge $e$ and its discovered endpoint $v \in V$. At the round when $v$ was discovered by a robot, that robot must have been performing a depth-next move because the depth of its anchor was less than or equal to the depth of $v$ which was adjacent to a dangling edge. Thus, the robot cannot 
have left $T(v)$ before the edge $e$ was visited. Consequently, that robot is still rooted at some ancestor $v_i$ of $v$, thus $e\in \cup_{i\in [k]}T(v_i)$.
\end{proof}

\section{A 2-player zero-sum game with balls in urns}\label{sec: game}
In this Section we introduce a 2-player zero-sum board game that will allow to prove the main Lemma used in the complexity analysis of \texttt{BFDN}. 

\paragraph{Game description.} At time $t\in \Nbb$, the board of the game is a list of $k$ integers $(n_1^t, \dots, n_k^t)$ that represent the load of $k$ urns 
with a total of $k$ balls. When the game starts at $t=0$, we have $n_i^0 = 1$ and at every instant we have $\sum_{i\in[k]}n_i^t=k$  and $n_i^t\geq 0$. At time $t$, player A (the adversary) chooses an urn $a_t \in [k]$ that is not empty, i.e. such that  $n_{a_t}^t\geq 1$, and then player B (the player) chooses an urn $b_t\in [k]$ and moves a ball from urn $a_t$ to urn $b_t$. At the beginning of time $t+1$, the board has thus changed by $n_{a_t}^{t+1} = n_{a_t}^{t}-1$ and $n_{b_t}^{t+1} = n_{b_t}^{t}+1$.

\paragraph{Goal of the game.} At a given time $t$, we denote by $U_t \subset \{1,\dots,k\}$ the set of urns that have never been selected by the adversary. At the start, $U_0 = \{1,\dots,k\}$ and after round $t$, $U_{t+1} = U_t\setminus\{a_t\}$. The game stops when all urns in $U_t$ contain at least $\Delta$ balls, i.e. $n_i^t\geq \Delta, \forall i \in U_t$. If $\Delta\geq k$,  the game thus stops when all urns have been chosen, i.e. $U_t=\emptyset$. The goal of player B is to end the game as soon as possible, the goal of the adversary is to play for as long as it can.

\paragraph{Strategy of the player.} We consider the following strategy used by player B against the adversary. At time $t$, player B chooses an urn $b_t$ that contains the least number of balls among the urns that have never been chosen by the adversary, i.e. $ b_t \in \arg\min_{i\in [k]\setminus\{a_1,\dots, a_t\}}n_i^t$.

We now state the main result of this section, of which our analysis of \texttt{BFDN} is a direct consequence.
\begin{theorem}\label{th:game}
If the player uses the strategy above, the game ends in at most $k\min\{\log(\Delta),\log(k)\}+k$ steps.    
\end{theorem}

\begin{proof}
The set $U_t$ does not increase with time. We denote its cardinality $u_t = |U_t|$. The strategy of player B implies that the difference between the number of balls in two distinct urns of $U_t$ is at most $1$. Consequently, denoting $N_t = \sum_{i\in U_t}n_i^t$ the total number of balls in urns of $U_t$, the number of balls in each urn of $U_t$ lies in $\{\ceil{\frac{N_t}{u_t}},\floor{\frac{N_t}{u_t}}\}$. The game thus stops as soon as $\frac{N_t}{u_t}\geq \Delta$ and the quantity $x_t := \Delta u_t - N_t$, must be positive as long as the game lasts.  
We distinguish two options for the adversary at any step $t$:
\begin{enumerate}
    \item The adversary chooses an urn $a_t$ that it previously chose ($a_t\not\in U_t$). In this case, $u_{t+1} = u_t$ and $N_{t+1}=N_t +1$. 
    Note that this option is available to the adversary only if some ball lies outside of $U_t$, i.e. if $N_t \leq k-1$.
    \item The adversary chooses an urn $a_t$ that it has never chosen before ($a_t\in U_t)$. In this case, $u_{t+1} = u_t-1$ and $N_{t+1}= N_t-n_{a_t}^{t}+1$. 
\end{enumerate}

\paragraph{Best response of the adversary.} For parameters $u,N \in\{0,\ldots,k\}$, we denote by $R(N,u)$ the largest number of steps that the game will still last after player B's move led to a configuration where $N_t=N$ and $u_t=t$ at any time $t$. Note that by the discussion above, this value is the same for all such configurations of the game. 
Clearly,
$$
\Delta u-N\le 0 \Rightarrow R(N,u)=0.
$$
Besides, in view of the options just listed, one has the following, assuming $\Delta u-N>0$:
\begin{equation}\label{eq:game_value}
\begin{array}{ll}
N<k &\Rightarrow R(N,u)=1 + \max\left(R(N-\lceil N/u\rceil+1,u-1),R(N-\lfloor N/u\rfloor+1,u-1), R(N+1,u)  \right),\\
N=k&\Rightarrow R(N,u)=1+\max\left(R(N-\lceil N/u\rceil+1,u-1),R(N-\lfloor N/u\rfloor+1,u-1)\right).
\end{array}
\end{equation}
We now establish the following
\begin{lemma}\label{lem:2ndlemme}
For any $(u,N)\in \{0,\ldots,k\}$, it holds that:

i) Function $M\to R(M,u)$ is non-increasing, and 

ii) The maximum in \eqref{eq:game_value} for $N<k$ is always achieved by $R(N+1,u)$.
\end{lemma}
\begin{proof}
For $u=0$, $R(M,u)\equiv 0$ and there is nothing to prove. Assume that the two properties i) and ii) hold for $v=u-1\ge 0$. We will show that ii) holds for $u$. Consider $N<k$. By the monotonicity assumption i),
$$
R(N-\lceil N/u\rceil+1,u-1)\ge R(N-\lfloor N/u\rfloor+1,u-1).
$$
Assume thus that the adversary moves first to configuration $(N-\lceil N/u\rceil+1,u-1)$.
By assumption ii) at rank $v$, its next best move is to configuration $(N-\lceil N/u\rceil+2,u-1)$. If alternatively the adversary had made a first move to $(N+1,u)$, it could then move to $(N+1-\lceil (N+1)/u\rceil +1, u-1)$. Now by the monotonicity assumption ii) this improves the adversary's reward if 
$N -\lceil N/u\rceil+2\ge N+1-\lceil (N+1)/u\rceil +1$, which is obviously true. We have thus established ii) at rank $u$. Monotonicity i) at rank $u$ readily follows, since we now have that $R(N+1,u)=R(N,u)-1$ if $\Delta u-N>0$. 
\end{proof}

By Lemma~\ref{lem:2ndlemme}, the adversary always chooses option 1. when it is available and chooses option 2. otherwise. Playing option 2. grants 
a budget to choose option 1. for another $\ceil{\frac{N_t}{u_t}}-1$ time steps.
This entirely determines the course of a game when the adversary responds optimally to the player's strategy. Note that in such game, $u_t$ is decremented by $1$ every $\ceil{\frac{k}{u_t}}$ steps. The game stops after $u_t \leq \frac{k}{\Delta}$, thus the last series of moves comes when $u_t = \ceil{\frac{k}{\Delta}}$. Assuming $\Delta\leq k$, the game then lasts a total time of $\ceil{\frac{k}{k}}+\ceil{\frac{k}{k-1}}+\dots \ceil{\frac{k}{\ceil{k/\Delta}}} \leq \sum_{h = \ceil{k/\Delta}}^{k}\left(\frac{k}{h}+1\right)\leq k\int_{k/\Delta}^k \frac{dx}{x} +k\leq k(\log(k)-\log(k/\Delta))+k = k\log(\Delta) + k$. Instead if $k<\Delta$, the game will stop after $u_t=1$ and the sum is thus bounded by $k\int_{1}^k \frac{dx}{x} +k\leq k\log(k)+k$.

The game therefore ends in at most $k\min\{\log(\Delta),\log(k)\}+k$ steps.
\end{proof}

\subsection{Connection to \texttt{BFDN}}
We now use the analysis of the two-player game to prove Lemma \ref{lem:main_lemma}.

\setcounter{theorem}{1}
\begin{lemma*}[Restated]
    In an execution of {\normalfont \texttt{BFDN}}, for any $d\in \{1,\dots,D-1\}$, the total number of reassignments of an anchor $v_i$ to a new value $v\in V$ at depth $\depth(v)=d$ is at most $k(\min\{\log(k),\log(\Delta)\}+2)$. 
\end{lemma*}
\setcounter{theorem}{4}
\begin{proof}

We start the proof of the Lemma by the following claim on \texttt{BFDN}.

\begin{claim}
At some round, if all anchors are at depth at most $d-1$, all nodes $v$ discovered at depth $d$ are in either of these (non-exclusive) situations: their sub-tree $T(v)$ is entirely discovered, or their sub-tree $T(v)$ hosts exactly one robot.
\end{claim}

\begin{proof}[Proof of claim 5]
Consider a discovered node $v$ at depth $d$ that contains a dangling edge in its sub-tree $T(v)$, we show that 
$T(v)$ hosts one robot. The dangling edge must have a discovered endpoint $v'\in T(v)$ that was attained by a robot performing depth-next moves. This robot cannot 
have left $T(v') \subset T(v)$ because $v'$ is still adjacent to a dangling edge, thus that robot is still in $T(v)$. At most one robot is in $T(v)$ because $v$ can only have been attained by a single robot, since all anchors are at depth $d-1$ or above.
\end{proof} 

The Lemma will result from the following reduction of the analysis of \texttt{BFDN} to the urns and balls game. We fix some depth $d\geq 1$ and bound  the number $N_d$ of times a robot is assigned a breadth-first move to some vertex at depth $d$ as follows. At the start of the earliest round when this happens, all anchors are at depth at most $d-1$. We consider the set $U$ of nodes  at depth $d$ that contain a robot in their  sub-tree. Obviously $|U|\leq k$ (in fact, $|U|\leq k-1$ because at least one robot must be at the root). Using the claim above, and the fact that there are no more dangling edges at depth $d-1$, we note that $U$ contains all nodes at depth $d$ that are adjacent to dangling edges, and thus all possible candidates for anchors at depth $d$. For each such candidate anchor, we formally re-anchor the robot exploring the corresponding sub-tree to this anchor (this does not change the algorithm's evolution).

We then increment counter $c$ at every instance of a robot re-anchoring, with possibly multiple increments within a single round. 

For counter increment to value $c$, we denote $a_c\in V$ the vertex to which the robot was previously anchored, and by $b_c \in U $ the vertex to which it is anchored next. Note all nodes in $\{a_1,\dots,a_c\}$ can no longer be adjacent to a dangling edge. We stop the increment the last time a robot is anchored at depth $d$, which happens when there does not remain any node at depth $d$ that is adjacent to some dangling edge.

Consider the counter value $C$ when for all nodes in $U$, either a robot  returning from it has reached the root, or at least $\Delta$ robots have been anchored at it. Then $C$ is the value of a run of the previous two-player game, initialized with one urn containing $k-u$ balls and $u$ urns each containing one ball, where $u=|U| \in \{0,\ldots,k-1\}$ and where player $B$ implements the balancing strategy. Indeed the re-anchoring strategy of \texttt{BFDN} balances the numbers of robots assigned per anchor. A direct adaptation of our analysis also holds for this modified initial condition of the game, yielding the upper bound on $C$ of $k(\minkdel+1)$. Once $C$ assignments at depth $d$ were made, at least $\Delta$ robots are assigned to nodes at depth $d$ that are still adjacent to a dangling edge. In the subsequent $d$ rounds \texttt{BFDN} can anchor each robot at most one last time before there is no more dangling edge at depth $d$. This yields the announced bound of  $k(\min(\log(k),\log(\Delta))+2)$ on $N_d$.
\end{proof}
\begin{remark}\label{rem:rem1}
Recall that upon counter increment to value  $c$, only nodes in $U \setminus\{a_1,\dots,a_{c}\}$ may be adjacent to a dangling edge. Consider again the urns-in-balls assignment rule $b_c = \arg\min_{v\in U \setminus\{a_1,\dots,a_{c}\}}n^c_v$, where $n_v^c$ denotes the number of nodes anchored at $v$ upon increment $c$, but where nodes in $U$ remain eligible as anchors until some robot has returned to the root from them. This is precisely the modified assignment rule we will need  in the variant of {\normalfont \texttt{BFDN}} in Section \ref{sec:commmunication}. Using again the above-mentioned formal re-anchoring of robots to nodes in $U$, the proof of Theorem \ref{th:game} entails that, for such modified assignment rule, a robot will have returned from all nodes of $U$ after at most $k(\min\{\log(k),\log(\Delta)\}+2)$ increments, after which  there are no more dangling edges at depth $d$.
\end{remark}

\section{Extensions of \texttt{BFDN} to alternative settings}\label{sec: other-settings}

We now consider three settings where a \texttt{BFDN} strategy enjoys non-trivial runtime guarantees.

\subsection{Restricted memory and communications}
\label{sec:commmunication}
In this section, we assume that robots are allowed to communicate with a central planner only when they are located at the root and that they have access to $\Delta + D\log(\Delta)$ bits of internal memory. We show that in this setting, a simple variant of \texttt{BFDN} achieves fast exploration.

Formally, we describe the setting as follows. At every node, the \textit{ports}, which are defined as the endpoints of the adjacent edges, are numbered from $1$ to $\Delta$ where $\Delta$ is the maximum degree. A node $v$ at depth $d\leq D$ is identified by the sequence of ports that leads to it from the root with $d\log_2(\Delta)$ bits. For every node distinct from the root, we assume that port number $1$ leads to the root. As before, robots operate in rounds. 

All robots arriving at the root at some round $t$ have their memory read and stored by the planner along with their identifier. The planner can then perform any computation and update the memory of the robots. For all robots arriving at some node $v$ distinct from the root at some round $t$, we assume that the robots can observe the list of all ports from which a robot has returned (these will be called ``finished ports''). Then the robot has two choices: \texttt{SELECT} a port number as next move, or use a local routine \texttt{PARTITION} with the following properties:
\begin{itemize}
    \item No two robots calling \texttt{PARTITION} at some node $v$ will be sent to the same port $j \geq 2$.
    \item If a robot calling \texttt{PARTITION} at node $v$ and round $t$ is sent to port $j\geq 1$, it means that \texttt{PARTITION}  has sent a robot to all ports $j'\geq j$ at round $t$ or before.
\end{itemize}

In this model, \texttt{BFDN} is implemented as follows. In a stack of $d$ port numbers (each represented by $\log_2(\Delta)$ bits) the central planner assigns to $\texttt{Robot}_i$ an anchor $v_i$ at depth $d$ that it will reach by unstacking port numbers and applying routine \texttt{SELECT}. When the robot reaches this node, the stack is empty and the robot will make consecutive calls to routine \texttt{PARTITION} that will eventually lead it back to the root. We ask that $\texttt{Robot}_i$ stores the finished port numbers of $v_i$ using its additional $\Delta$ bits of memory. This information will be used by the central planner to update its candidates for future anchors, as specified in Algorithm~\ref{alg:BFDN_planner}.

\begin{algorithm}
\caption{\texttt{BFDN} ``Breadth-First Depth-Next'' (with central planner at the root)}\label{alg:BFDN_planner}
\begin{algorithmic}[1]
\Require At most $k$ robots arriving at the root at some round.
\Ensure Assigns a node $v$, represented by a sequence of port numbers, to each robot.
\State $d = $ working depth ; 
\State $A = $ list of anchors at depth $d$ ;
\State $R = $ nodes of $A$ from which a robot has returned ;
\State $A' = $ list of children of nodes in $A$ ;
\State $R' = $ nodes of $A'$ from which a robot has returned ;
\State \textbf{Read memory} of returning robots and update $R,R'$ accordingly.\label{line:deduction}\Comment{see proof.}
\If{$A \setminus R = \emptyset$}
\State $A \gets A'\setminus R'$ \Comment{contains at most $k$ elements.}
\State $R, A', R' \gets \emptyset$
\State $d \gets d+1$ 
\EndIf
\State \textbf{if} $A \setminus R = \emptyset$ \textbf{then} exploration is finished and robots wait at the root to be joined by the other returning robots \textbf{else} anchor robots to nodes of $A\setminus R$, ensuring that the total number of robots anchored at these nodes differs by at most one. 
\end{algorithmic}
\end{algorithm}

\begin{proposition}
    Under the present communication model with a central planner at the root, the modified {\normalfont \texttt{BFDN}} achieves exploration in at most $\frac{2n}{k} + D^2(\min\{\log(k),\log(\Delta)\}+2)$ rounds.
\end{proposition}

\begin{proof}
    The proof is similar to that of Theorem \ref{th:BFDN_main}. All claims 1-5 remain valid for this variant of the algorithm with essentially the same arguments. The lemma is slightly adapted as in Remark~\ref{rem:rem1} because an anchor remains eligible until a robot assigned to it has returned to the root (it implies that this anchor is no longer adjacent to a dangling edge). Algorithm \ref{alg:BFDN_planner} specifies how the central planner uses returning robots to update $A\setminus R$, the set of eligible anchors at the working depth $d$. When $A\setminus R=\emptyset$, a robot has returned from all anchors at depth $d$ and $d$ is incremented. The planner also keeps track of $A'\setminus R'$, which contains the children of $A$ that may be adjacent to a dangling edge, or equivalently the ports of $A$ that are not known to be finished. For this, we use the fact that robots store the list of finished ports at their anchor with $\Delta$ bits of memory. 
\end{proof}
\begin{remark}
        The present model encompasses the more classical ``local communication'' model where robots with unbounded memory may communicate when they are located at the same node at the same round (see e.g. \cite{dereniowski2015fast}). Indeed, under these assumptions, one may leave a specific robot at the root to play the role of the central planner, and also leave a robot stationned at all anchors to keep track of the finished ports. \texttt{BFDN} thus leads to a $\frac{4n}{k-1}+D^2(\min\{\log(k),\log(\Delta)\}+2)$ guarantee for this model.
    \end{remark}

\subsection{Adversarial robot break-downs}\label{sec:adversarial} So far we assumed that all moving robots traverse exactly one edge per time-step. We relax this assumption in the present Section, assuming instead that some adversary decides at each time-step and for each robot whether the robot actually moves, or instead incurs a break-down, being stalled at its current location.

Our aim is again to
to discover the tree in as few moves as possible. However we no longer require that the robots return to the root at the end of exploration, because the adversary could decide to break down some robot indefinitely.

At each round $t\in \Nbb$,  robot $i$ is allowed to make a move if some variable $M_{ti}=1$ whereas it is blocked at its current position if $M_{ti} = 0$. For this adversarial model, we assume that  $\mathbb{M}=(M_{ti})_{t\in\Nbb,i\in[k]}$ is an arbitrary sequence of binary values. We denote the average distance travelled by the robots $A(\mathbb{M})$ which equals $A(\mathbb{M}) =\frac{1}{k}\sum_{t\in \Nbb}\sum_{i\in [k]}M_{ti}$.

For this setting, we consider \texttt{BFDN} as specified in Algorithm \ref{alg:BFDN_algo}, with the minor modification that at each round $t$ the only robots taking part in the assignment process are those which are allowed to move. More precisely, we replace the \texttt{for} loop of Algorithm \ref{alg:BFDN_algo} (\textbf{for} $i\in \{1,\dots,k\}$ \textbf{do}) with an iteration over all robots that may move (\textbf{for} $i\in \{i : M_{ti}=1\}$ \textbf{do}). This modification is introduced to ensure that when multiple robots are at the same location, blocked robots do not prevent unblocked robots from traversing dangling edges. 
We then have the following 

\begin{proposition}
    Under the model with adversarial breakdowns, for any sequence of allowed moves $\mathbb{M}\in \{0,1\}^{\Nbb\times[k]}$ satisfying $A(\mathbb{M})\geq \frac{2n}{k}+D^2(\log(k)+2)$ all edges will be visited by the above modification of {\normalfont \texttt{BFDN}}.
\end{proposition}
\begin{proof} Again, the proof is very similar to that of Theorem \ref{th:BFDN_main} and all claims \texttt{1-5} all naturally adapt to this setting. As an example, we adapt the third claim as follows.

\setcounter{claim}{2}
\begin{claim}[Restated]
Consider a sequence of moves by some $\texttt{Robot}_i$ starting at the root with the assignment of an anchor $v$ of depth $\depth(v)=d$ and ending with the return of the robot to the root. We denote by $\T_x$ the number of moves that $\texttt{Robot}_i$ was allowed to perform during this sequence by the adversary. In this sequence, $\texttt{Robot}_i$ has explored exactly $( \T_x-2d)/2$ dangling edges.
\end{claim}
The adversarial nature of the urns and balls game described in Section \ref{sec: game} also makes it applicable to the present setup, and the main lemma straightforwardly holds except for the $\log(\Delta)$ guarantee. Indeed, the adversary could choose to block all robots at a specific anchor until all $k$ robots reach that anchor, which happens after at most $k(\log(k)+1)$ anchor assignements.
\end{proof}
\begin{remark}
Other adversarial settings could be considered, for instance with an adversary that observes the moves that the robots have selected before choosing which robots to block. Another extension of interest would consist in relaxing the slotted time assumption to consider instead continuous time evolution, which could capture more realistic scenarios, with varying robot speeds. 

\end{remark}

\subsection{Collaborative exploration of non-tree graphs}
The algorithm \texttt{BFDN} described above can be executed on a graph, with a minor modification: any robot that lands on a node discovered earlier by another robot should go back from where it came and ``close'' the corresponding edge (this edge will never be used again). A similar technique was already proposed by \cite{brass2011multirobot} to adapt the algorithm of \cite{fraigniaud2006collective} to graphs. Unfortunately, without further assumption, the guarantees of \texttt{BFDN} do not generalize to graphs with $n$ \textit{edges} and \textit{radius} $D$, where the radius is defined as the maximum distance between a node and the origin of the robots.

We therefore make the additional assumption that at any given node, a robot knows its distance to the origin in the underlying graph. Though restrictive, this assumption holds in some contexts of interest. It is for instance satisfied for the exploration of grid graphs with rectangular obstacles considered in  \cite{ortolf2012online} because in such graphs, the distance to the origin of any node with coordinates $(i,j)\in \Nbb^2$ is always exactly equal to the so-called Manhattan distance $i+j$.

In that context, consider the algorithm \texttt{BFDN}, with the modification that a robot exploring some edge $e$ for the first time will backtrack and ``close'' this edge if either of these two conditions is satisfied: (1) $e$ led to a node that is already discovered (2) $e$ led to a node that is not strictly further to the origin than its first endpoint. We then have the following 
\begin{proposition}
Given a graph $G=(V,E)$ with $n$ edges, diameter $D$ and maximum degree $\Delta$, assume that the $k$ robots are aware at all times of their distance to the origin and implement the above variant of {\normalfont \texttt{BFDN}}. Then collaborative exploration of the graph is completed in at most $\frac{2n}{k} + D^2(\min\{\log(\Delta),\log(k)\}+2)$ rounds.
\end{proposition}

\begin{proof}
    It is clear that at the end of the execution of this algorithm, the edges that have never been closed form a breadth-first tree of the graph of depth $D$. This tree is explored efficiently by our algorithm while the edges that were closed were traversed at most twice by a single robot (or once by two robots, each coming from both endpoints, that will swap their identities). This leads to a total runtime of at most $\frac{2n}{k} + D^2(\min\{\log(\Delta),\log(k)\}+2)$.
\end{proof}


\section{Recursive Algorithms for Improved Dependence on Depth $D$}
\label{sec:trading}
In this Section we develop a general recursive construction of so-called {\em anchor-based algorithms} which,  applied to \texttt{BFDN}, yields the following result. It can be seen as a generalization of Theorem~\ref{th:BFDN_main} as, for $\ell=1$, it provides the same upper-bound up to a factor four. 
\begin{theorem}\label{th:BFDN-ell}
  For any integer $\ell\ge 1$, \texttt{BFDN}$_\ell$, an associated recursive version of \texttt{BFDN}, explores a tree with $n$ nodes, depth $D$, maximum degree $\Delta$ with $k$ robots in $\frac{4n}{k^{1/\ell}} + 2^{\ell+1}(\ell+1+\min\set{\log(\Delta),\log(k)/\ell})\, D^{1+1/\ell}$ rounds.
\end{theorem}

To describe our recursive construction we need
the following definitions.
Given a node $v$ in a tree $T$, $P_T[v]$ denotes the path from $v$ to the root of $T$, and  $P_T(v)=P_T[v]\setminus\{v\}$.
Given two nodes $u,v$ in a tree $T$,  $\lca_T(u,v)$ denotes their lowest common ancestor in $T$.
We say that a discovered node is \emph{open} as long as it has at least one dangling adjacent edge. We say that it is \emph{closed} as soon as a robot has traversed its last dangling edge.
Note that open nodes are the parents of dangling edges.
We decompose the exploration of an edge into two edge events as follows.
An \emph{edge event} occurs when a robot traverses an edge from parent to child for the first time, or when a robot traverses an edge from child to parent for the first time. There are thus at most $2(n-1)$ edge events in any exploration. 
Edges for which only one event has occurred are said to be \emph{half explored}.

\paragraph{Anchor-based algorithm.} 
Given $k$ robots, an activity parameter $k^*\in [k]$, and a depth $d$, an \emph{anchor-based} algorithm $\Acal(k^*,k,d)$ is by definition an exploration algorithm by $k$ robots meeting the following requirements. Each robot is in one of the two states \emph{active} or \emph{inactive}. Each active robot $i$ is assigned to a node $v_i$ of the tree called its \emph{anchor}.
The algorithm must explore the tree so as to bring anchors at depth $d$ while maintaining a list of invariants. The full list of so-called ``Anchor-based invariants'' is given in Appendix~\ref{app:anchor-based}. It mainly includes a variant of Claim~4 called \emph{Open Node Coverage} which specifies that all open nodes must always be in $\cup_{i\in A}T(v_i)$ where $A$ is the set of active robots. Other invariants mainly specify properties of the positions of the robots with respect to the partially explored tree and ensure that we can start an execution of an anchor-based algorithm after having interrupted the execution of another anchor-based algorithm.

Initially, the algorithm starts from any partially explored tree, with all robots active and anchored at the root. Robots must be in so-called \emph{Parallel DFS Positions}, a requirement ensuring that all invariants are initially satisfied (see Appendix~\ref{app:anchor-based}).
Active robots are allowed to move and explore the tree while inactive robots must be at depth at most $d$ and wait. We distinguish two phases in the execution of the algorithm. As long as some anchor is at depth less than $d$ or is not closed, we say that the algorithm runs \emph{shallow}. During this first ``shallow'' phase, the algorithm must have at least $k^*$ active robots at all rounds. When all anchors are at depth $d$ and are all closed, we say that the algorithm runs \emph{deep}. In this second ``deep'' phase, it is required that all active robots trigger an edge event at each round. However, the number of active robots may get below $k^*$ during that phase. At any round, the algorithm may turn a robot into inactive or active as long as the requirements for the two phases are met. Finally, the algorithm can terminate when all robots are inactive. The Open Node Coverage invariant implies that the tree is then completely discovered (see Appendix~\ref{app:anchor-based}).

\paragraph{Divide depth functor.} We now define the \emph{divide depth functor}  $\Dcal$, a map that takes an anchor-based algorithm and transforms it into another anchor-based algorithm as follows. Given an anchor-based algorithm $\Acal(k^*,k',d')$, a number $n_{team}$ of teams and a number $n_{iter}$ of iterations, we construct the exploration algorithm $\Dcal[\Acal(k^*,k',d'); n_{team}; n_{iter}]$ for terminating the exploration of a partially explored tree. 
It uses $k=n_{team}k'$ robots for  exploring the tree up to depth $d=n_{iter}d'$ in $n_{iter}$ iterations where each iteration makes anchors progress $d'$ deeper. More precisely, the $i$-th iteration runs parallel instances of $\Acal(k^*,k',d')$ in at most $n_{team}$ sub-trees rooted at nodes with depth $(i-1)d'$. We assume that the previous iteration has terminated with a set $R$ of at most $k^*\le n_{team}$ anchors at depth $(i-1)d'$. Relying on the Open Node Coverage invariant, we then restrict the exploration to the sub-trees rooted in $R$. Robots are thus partitioned into $n_{team}$ teams of $k'$ robots each. Each node $r\in R$ is taken in charge by a distinct team which runs an instance $\Acal_r(k^*,k',d')$ of $\Acal(k^*,k',d')$ on $T(r)$. When $|R|<n_{team}$, all robots in unassigned teams are inactive and wait at their position until the end of the current iteration. All other teams explore in parallel their sub-trees. We interrupt all running instances simultaneously when the overall number of active robots gets below $k^*$ so that we can use their anchors as roots in the next iteration. As any single instance has activity parameter $k^*$ this cannot happen until all anchors are at depth $d'$ in each sub-tree, that is depth $i\cdot d'$ in $T$.
After $n_{iter}$ iterations, this guarantees that all nodes up to depth $d$ have been closed and that exploration finally continues in at most $k^*$ sub-trees rooted at depth $d$. See Appendix~\ref{sec:pseudocode_div_depth} for a formal description of the resulting anchor-based algorithm $\Bcal(k^*,k,d)=\Dcal[\Acal(k^*,k',d'); n_{team}; n_{iter}]$.

We say that an anchor-based $\Acal(k^*,k,d)$ algorithm has \emph{$f$-shallow efficiency} for parameter $f$ if it triggers at least $k^*(\T - f)$ edge events when running shallow during $\T$ rounds where parameter $f$ may depend on $k$ and $d$.
We then have the following 

\begin{proposition}\label{prop:divide-depth}
  Given an anchor-based algorithm $\Acal(k^*,k',d')$, integers $n_{team}\ge k^*$ and $n_{iter}\ge 1$,
  $\Dcal[\Acal(k^*,k',d'); n_{team}; n_{iter}]$ is correct and it is an anchor-based exploration algorithm $\Bcal(k^*,k,d)$ for $k=n_{team}k'$ robots with depth $d=n_{iter}d'$.
 If moreover $\Acal(k^*,k',d')$ has $f'$-shallow efficiency,  then $\Dcal[\Acal(k^*,k',d'); n_{team}; n_{iter}]$ has $f$-shallow efficiency with $f=n_{iter}f'+n_{iter}^2d'=n_{iter}(f'+d)$.
\end{proposition}
Its proof is deferred to Appendix~\ref{sec:pseudocode_div_depth}.
The reason for $f$-shallow efficiency is the following.
Consider the $i$-th iteration of $\Dcal_{\Acal,k',d'}(k^*,k,d)$. Moving robots towards their associated root takes $2(i-1)d'$ rounds. Now, count the number $\T^1$ of rounds where at least one of the instances has not run deep. As such an instance has run shallow during $\T^1$ rounds, it has triggered at least $k^*(\T^1 - f')$ edge events by $f'$-shallow efficiency of $\Acal(k^*,k,d)$. During the remaining $\T^{2}$ rounds of the iteration, all instances run deep. As this continues as long as $k^*$ robots or more are active, at least $k^*$ edge events are triggered per round, that is $k^*\T^2$ or more in total. Letting $\T_i=2(i-1)d'+\T^1+\T^2$ denote the number of rounds spent in the $i$th iteration, the number of edge events triggered during that iteration is thus at least $k^*(\T_i-f'-2(i-1)d')$. The algorithm runs shallow during the $n_{iter}$ iterations which last overall $\T=\sum_{i=1}^{n_{iter}}\T_i$. By summation, we get that it then triggers at least $k^*(\T-n_{iter}f'-n_{iter}^2d')$ edge events as 
$\sum_{i=1}^{n_{iter}}(i-1)<n_{iter}^2/2$.

Our first candidate for applying the divide depth functor is the following 
variant of \texttt{BFDN}.

\paragraph{BFDN}
In the sequel, we call $\DFBNmath_1(k,k,d)$ the modification of Algorithm~\ref{alg:BFDN_algo} where the procedure \texttt{Reanchor} is modified for assigning anchors at depth at most $d$. Precisely, we replace Line~\ref{line:anchorset} with:
$$
A = \{v \in V ~~ s.t. ~~ v \text{ is adjacent to some unexplored edge and } \delta(v) \text{ is minimal and } \delta(v)\le d\}.
$$
Note that this modification implies that when there are no more dangling edges at depth at most $d$, robots start to be anchored to the root and are then considered as inactive. 
Note that according to Claim~5 for depth $d+1$, there still remains exactly one robot in each sub-tree rooted at depth $d+1$ which is not entirely discovered. These robots remain active until they have completely explored their sub-tree. 
$\DFBNmath_1(k,k,d)$ thus terminates only when the tree has been fully discovered. 
We also slightly modify the anchoring of robots: when a robot $i$ is anchored at $v_i$ it might happen that there are no more dangling edges at depth $\depth(v_i)$ or less thanks to the exploration of other robots. If this happens when $v_i\in P(u_i)$ and $\depth(v_i)<d$, we re-anchor robot $i$ at the children of $v_i$ in $P[u_i]$.
This modification does not change the movements of robot $i$ as it is then in a sequence of depth-next moves and will go up when reaching $v_i$ anyway. However, this modification will ensure the preservation of the Partial Exploration invariant defined in Appendix~\ref{app:anchor-based}. It also implies that when there are no more dangling edges at depth at most $d$, all anchors are then at depth $d$.

One can then easily check that $\DFBNmath_1(k,k,d)$ is an anchor-based algorithm. For example, the Open Node Coverage invariant is shown as Claim~4; see Appendix~\ref{app:anchor-based} for more details.

We also note that $\DFBNmath_1(k,k,d)$ has $c_1(k)d^2$-shallow efficiency where $c_1(k)=\minkdel + 2$.
Indeed, $\DFBNmath_1(k,k,d)$ runs exactly as Algorithm~\ref{alg:BFDN_algo} as long as there are dangling edges at depth at most $d$, that is as long as the algorithm is running shallow. If this phase lasts $\T$ rounds, it triggers at least $k(\T-c_1(k)d^2)$ edge events. The proof is similar to that of Theorem~\ref{th:BFDN_main} using Lemma~\ref{lem:main_lemma} with the slight subtlety that we count edge events. The reason is that when starting from a partially explored tree where robots are in Parallel DFS Positions, the moves when robots go up still trigger edge events although no new edge may be discovered. 

\paragraph{The $\DFBNmath_\ell(k^*,k,d)$ anchor-based algorithm.} 

We construct recursively a series of algorithms $\DFBNmath_\ell(k^{1/\ell},k,d)$ for $\ell \ge 1$ as follows.
Assuming that $k$ and $d$ are both $\ell$-th powers of integers, we define for $\ell\ge 2$ the algorithm $\DFBNmath_\ell(k^*,k,d) := \Dcal[\DFBNmath_{\ell-1}(k^*,k/n_{team},d/n_{iter}); n_{team}; n_{iter}]$ with $k^*=n_{team}=k^{1/\ell}$ and $n_{iter}=d^{1/\ell}$. We let $k'=k/n_{team}=k^{(\ell-1)/\ell}$ and $d'=d/n_{iter}=d^{(\ell-1)/\ell}$ denote the parameters used for $\DFBNmath_{\ell-1}$. Note that $k'$ and $d'$ are both $(\ell-1)$-th powers of integers and recursive calls all have integer-valued parameters. The activity parameter of instances $\DFBNmath_{\ell-1}(k^*,k',d')$ indeed satisfies $(k')^{1/(\ell-1)}=k^{1/\ell}=k^*$.  As we use $n_{team}=k^*$, we indeed respect the constraint $k^*\le n_{team}$.
We can bound its shallow efficiency according to the following statement:
\begin{lemma}\label{lem:BFDN-ell-shallow-efficiency}
  Given an integer $\ell\ge 2$, two integers $k$ and $d$ that are both $\ell$th powers of integers, $\DFBNmath_\ell(k^{1/\ell},k,d)$ is $c_\ell(k)d^{1+1/\ell}$-shallow efficient with $c_\ell(k)=c_1(k^{1/\ell})+\ell-1$.
\end{lemma}
\begin{proof}
As $\DFBNmath_1(k^{1/\ell},k^{1/\ell},d^{1/\ell})$ is $c_1(k^{1/\ell})d^{2/\ell}$-shallow efficient, Proposition~\ref{prop:divide-depth} implies by induction 
that $\DFBNmath_j(k^{1/\ell},k^{j/\ell},d^{j/\ell})$ is $(c_1(k^{1/\ell})+j-1)d^{(j+1)/\ell}$-shallow efficient for $j=2,\ldots,\ell$.
\end{proof}

Algorithm \texttt{BFDN}$_\ell$ in Theorem~\ref{th:BFDN-ell} is then defined as
\begin{definition}
In case $k$ is the $\ell$-th power of some integer, consider the sequence of depths $d_j=2^{j\ell}$ for $j=1,2,\ldots$ Algorithm \texttt{BFDN}$_\ell$ consists in running $\DFBNmath_\ell(k^{1/\ell},k,d_1)$, interrupting it right after its last iteration (without running deep further), then running $\DFBNmath_\ell(k^{1/\ell},k,d_2)$ with the current robot positions and anchor assignments until its last iteration finishes, and so on. When running $\DFBNmath_\ell(k^{1/\ell},k,d_j)$ with  $j=\ceil{\frac{\log_2 D}{\ell}}$, all anchors reach depth $D$ and the algorithm terminates. If $k$ is not an integer to the power $\ell$, we only use $K=\floor{k^{1/\ell}}^\ell\le k$ and apply the previous construction. 
\end{definition}

\begin{proof} (of Theorem~\ref{th:BFDN-ell})
Assume first that $k$ is the $\ell$-th power of some integer. 
In a run of \texttt{BFDN}$_\ell$, denote by $\T_j$ the number of rounds that the call to $\DFBNmath_\ell(k^{1/\ell},k,d_j)$ lasts. This call triggers at least $k^{1/\ell}(\T_j-c_\ell(k)d_j^{1+1/\ell})$ edge events by applying Lemma~\ref{lem:BFDN-ell-shallow-efficiency}. We can thus bound the overall running time $\T=\sum_{j=1}^{\ceil{(\log_2 D)/\ell}}\T_j$ by summing over all calls:
$
  2n \ge  k^{1/\ell} \left( \T - c_\ell(k)\sum_{j=1}^{\ceil{(\log_2 D)/\ell}} d_j^{1+1/\ell} \right)
$. 
As we have 
$
  \sum_{j=1}^{\ceil{(\log_2 D)/\ell}}d_j^{1+1/\ell}
  = \sum_{j=1}^{\ceil{(\log_2 D)/\ell}} 2^{(\ell+1)j}
  \le \frac{2^{(\ell+1) ((\log_2 D)/\ell+2)} - 1}{2^{\ell+1}-1}
  \le 2^{\ell+1} D^{1+1/\ell},
$
we obtain
\begin{align*}
\T \le \frac{2n}{k^{1/\ell}} + 2^{\ell+1}c_\ell(k)D^{1+1/\ell}.
\end{align*}
\smallskip
For arbitrary $k$, with $K=\floor{k^{1/\ell}}^\ell$, using $K^{1/\ell}\ge k^{1/\ell}/2$, we obtain  a time bound of
\begin{equation*}
  \T \le \frac{4n}{k^{1/\ell}} + 2^{\ell+1}(\ell-1+c_1(k^{1/\ell}))D^{1+1/\ell},
\end{equation*}
yielding the runtime bound announced in Theorem~\ref{th:BFDN-ell} since $c_1(k^{1/\ell})=2+\min\set{\log(\Delta),\log(k)/\ell}$.
\end{proof} 

\section*{Acknowledgement}
RC thanks Maxime Cartan for careful reading and for implementing \texttt{DFBN} in Python. This work was supported by ANR-19-P3IA-0001 (PRAIRIE 3IA Institute) and ANR Tempogral ANR-22-CE48-0001.

\bibliography{biblio}

\begin{thebibliography}{14}
\providecommand{\natexlab}[1]{#1}
\providecommand{\url}[1]{\texttt{#1}}
\expandafter\ifx\csname urlstyle\endcsname\relax
  \providecommand{\doi}[1]{doi: #1}\else
  \providecommand{\doi}{doi: \begingroup \urlstyle{rm}\Url}\fi

\bibitem[Fraigniaud et~al.(2006)Fraigniaud, Gasieniec, Kowalski, and
  Pelc]{fraigniaud2006collective}
Pierre Fraigniaud, Leszek Gasieniec, Dariusz~R. Kowalski, and Andrzej Pelc.
\newblock Collective tree exploration.
\newblock \emph{Networks}, 48\penalty0 (3):\penalty0 166--177, 2006.
\newblock \doi{10.1002/net.20127}.
\newblock URL \url{https://doi.org/10.1002/net.20127}.

\bibitem[Brass et~al.(2011)Brass, Cabrera{-}Mora, Gasparri, and
  Xiao]{brass2011multirobot}
Peter Brass, Flavio Cabrera{-}Mora, Andrea Gasparri, and Jizhong Xiao.
\newblock Multirobot tree and graph exploration.
\newblock \emph{{IEEE} Trans. Robotics}, 27\penalty0 (4):\penalty0 707--717,
  2011.
\newblock \doi{10.1109/TRO.2011.2121170}.
\newblock URL \url{https://doi.org/10.1109/TRO.2011.2121170}.

\bibitem[Disser et~al.(2017)Disser, Mousset, Noever, Skoric, and
  Steger]{disser2017general}
Yann Disser, Frank Mousset, Andreas Noever, Nemanja Skoric, and Angelika
  Steger.
\newblock A general lower bound for collaborative tree exploration.
\newblock In Shantanu Das and S{\'{e}}bastien Tixeuil, editors,
  \emph{Structural Information and Communication Complexity - 24th
  International Colloquium, {SIROCCO} 2017, Porquerolles, France, June 19-22,
  2017, Revised Selected Papers}, volume 10641 of \emph{Lecture Notes in
  Computer Science}, pages 125--139. Springer, 2017.
\newblock \doi{10.1007/978-3-319-72050-0\_8}.
\newblock URL \url{https://doi.org/10.1007/978-3-319-72050-0\_8}.

\bibitem[Higashikawa et~al.(2014)Higashikawa, Katoh, Langerman, and
  Tanigawa]{higashikawa2014online}
Yuya Higashikawa, Naoki Katoh, Stefan Langerman, and Shin{-}ichi Tanigawa.
\newblock Online graph exploration algorithms for cycles and trees by multiple
  searchers.
\newblock \emph{J. Comb. Optim.}, 28\penalty0 (2):\penalty0 480--495, 2014.
\newblock \doi{10.1007/s10878-012-9571-y}.
\newblock URL \url{https://doi.org/10.1007/s10878-012-9571-y}.

\bibitem[Dynia et~al.(2006{\natexlab{a}})Dynia, Kutylowski, auf~der Heide, and
  Schindelhauer]{dynia2006smart}
Miroslaw Dynia, Jaroslaw Kutylowski, Friedhelm~Meyer auf~der Heide, and
  Christian Schindelhauer.
\newblock Smart robot teams exploring sparse trees.
\newblock In Rastislav Kralovic and Pawel Urzyczyn, editors, \emph{Mathematical
  Foundations of Computer Science 2006, 31st International Symposium, {MFCS}
  2006, Star{\'{a}} Lesn{\'{a}}, Slovakia, August 28-September 1, 2006,
  Proceedings}, volume 4162 of \emph{Lecture Notes in Computer Science}, pages
  327--338. Springer, 2006{\natexlab{a}}.
\newblock \doi{10.1007/11821069\_29}.
\newblock URL \url{https://doi.org/10.1007/11821069\_29}.

\bibitem[Dynia et~al.(2007)Dynia, Lopuszanski, and
  Schindelhauer]{dynia2007robots}
Miroslaw Dynia, Jakub Lopuszanski, and Christian Schindelhauer.
\newblock Why robots need maps.
\newblock In Giuseppe Prencipe and Shmuel Zaks, editors, \emph{Structural
  Information and Communication Complexity, 14th International Colloquium,
  {SIROCCO} 2007, Castiglioncello, Italy, June 5-8, 2007, Proceedings}, volume
  4474 of \emph{Lecture Notes in Computer Science}, pages 41--50. Springer,
  2007.
\newblock \doi{10.1007/978-3-540-72951-8\_5}.
\newblock URL \url{https://doi.org/10.1007/978-3-540-72951-8\_5}.

\bibitem[Dereniowski et~al.(2013)Dereniowski, Disser, Kosowski, Pajak, and
  Uznanski]{dereniowski2015fast}
Dariusz Dereniowski, Yann Disser, Adrian Kosowski, Dominik Pajak, and
  Przemyslaw Uznanski.
\newblock Fast collaborative graph exploration.
\newblock In Fedor~V. Fomin, Rusins Freivalds, Marta~Z. Kwiatkowska, and David
  Peleg, editors, \emph{Automata, Languages, and Programming - 40th
  International Colloquium, {ICALP} 2013, Riga, Latvia, July 8-12, 2013,
  Proceedings, Part {II}}, volume 7966 of \emph{Lecture Notes in Computer
  Science}, pages 520--532. Springer, 2013.
\newblock \doi{10.1007/978-3-642-39212-2\_46}.
\newblock URL \url{https://doi.org/10.1007/978-3-642-39212-2\_46}.

\bibitem[Ortolf and Schindelhauer(2014)]{OrtolfS14}
Christian Ortolf and Christian Schindelhauer.
\newblock A recursive approach to multi-robot exploration of trees.
\newblock In Magn{\'{u}}s~M. Halld{\'{o}}rsson, editor, \emph{Structural
  Information and Communication Complexity - 21st International Colloquium,
  {SIROCCO} 2014, Takayama, Japan, July 23-25, 2014. Proceedings}, volume 8576
  of \emph{Lecture Notes in Computer Science}, pages 343--354. Springer, 2014.
\newblock \doi{10.1007/978-3-319-09620-9\_26}.
\newblock URL \url{https://doi.org/10.1007/978-3-319-09620-9\_26}.

\bibitem[Dynia et~al.(2006{\natexlab{b}})Dynia, Korzeniowski, and
  Schindelhauer]{dynia2006power}
Miroslaw Dynia, Miroslaw Korzeniowski, and Christian Schindelhauer.
\newblock Power-aware collective tree exploration.
\newblock In Werner Grass, Bernhard Sick, and Klaus Waldschmidt, editors,
  \emph{Architecture of Computing Systems - {ARCS} 2006, 19th International
  Conference, Frankfurt/Main, Germany, March 13-16, 2006, Proceedings}, volume
  3894 of \emph{Lecture Notes in Computer Science}, pages 341--351. Springer,
  2006{\natexlab{b}}.
\newblock \doi{10.1007/11682127\_24}.
\newblock URL \url{https://doi.org/10.1007/11682127\_24}.

\bibitem[Brass et~al.(2014)Brass, Vigan, and Xu]{brass2014improved}
Peter Brass, Ivo Vigan, and Ning Xu.
\newblock Improved analysis of a multirobot graph exploration strategy.
\newblock In \emph{13th International Conference on Control Automation Robotics
  {\&} Vision, {ICARCV} 2014, Singapore, December 10-12, 2014}, pages
  1906--1910. {IEEE}, 2014.
\newblock \doi{10.1109/ICARCV.2014.7064607}.
\newblock URL \url{https://doi.org/10.1109/ICARCV.2014.7064607}.

\bibitem[Aldous and Fill(1995)]{aldous1995reversible}
David Aldous and James Fill.
\newblock Reversible markov chains and random walks on graphs (monograph),
  1995.

\bibitem[Broder et~al.(1989)Broder, Karlin, Raghavan, and
  Upfal]{broder1989trading}
Andrei~Z. Broder, Anna~R. Karlin, Prabhakar Raghavan, and Eli Upfal.
\newblock Trading space for time in undirected s-t connectivity.
\newblock In David~S. Johnson, editor, \emph{Proceedings of the 21st Annual
  {ACM} Symposium on Theory of Computing, May 14-17, 1989, Seattle, Washington,
  {USA}}, pages 543--549. {ACM}, 1989.
\newblock \doi{10.1145/73007.73059}.
\newblock URL \url{https://doi.org/10.1145/73007.73059}.

\bibitem[Alon et~al.(2008)Alon, Avin, Kouck{\'{y}}, Kozma, Lotker, and
  Tuttle]{alon2008many}
Noga Alon, Chen Avin, Michal Kouck{\'{y}}, Gady Kozma, Zvi Lotker, and Mark~R.
  Tuttle.
\newblock Many random walks are faster than one.
\newblock In Friedhelm~Meyer auf~der Heide and Nir Shavit, editors,
  \emph{{SPAA} 2008: Proceedings of the 20th Annual {ACM} Symposium on
  Parallelism in Algorithms and Architectures, Munich, Germany, June 14-16,
  2008}, pages 119--128. {ACM}, 2008.
\newblock \doi{10.1145/1378533.1378557}.
\newblock URL \url{https://doi.org/10.1145/1378533.1378557}.

\bibitem[Ortolf and Schindelhauer(2012)]{ortolf2012online}
Christian Ortolf and Christian Schindelhauer.
\newblock Online multi-robot exploration of grid graphs with rectangular
  obstacles.
\newblock In Guy~E. Blelloch and Maurice Herlihy, editors, \emph{24th {ACM}
  Symposium on Parallelism in Algorithms and Architectures, {SPAA} '12,
  Pittsburgh, PA, USA, June 25-27, 2012}, pages 27--36. {ACM}, 2012.
\newblock \doi{10.1145/2312005.2312010}.
\newblock URL \url{https://doi.org/10.1145/2312005.2312010}.

\end{thebibliography}

\newpage

\appendix
\section{Comparisons between Algorithms \texttt{CTE}, \texttt{Yo}* and \texttt{BFDN}}\label{sec: conclusion}
We provided in Figure \ref{fig:CTEvsBMRDFS} a picture of how  \texttt{BFDN} compares in terms of runtime with other state-of-the art algorithms for collaborative tree exploration. The regions in the picture are defined up to multiplicative constants that only depend on $k$. We decided to include in this picture only algorithms for which guarantees are achieved irrespective of assumptions on the tree structure and that are the most efficient for some values of size $n$ and depth $D$. This leaves us only with three algorithms: the original ``collaborative tree exploration'' \texttt{CTE} algorithm of \cite{fraigniaud2006collective} with runtime  $\Ocal(\frac{n}{\log(k)}+D)$, the  recursive \texttt{Yo}* algorithm of \cite{OrtolfS14} with runtime  $\Ocal(2^{\Ocal(\sqrt{\log D \log\log k})}\log k(\log n+\log k)(n/k+D))$, which we reduced to smaller quantities to simplify the picture, and \texttt{BFDN} with runtime $2n/k+D^2\log(k)$ as well as its recursive variant $\texttt{BFDN}_\ell$. 

Figure \ref{fig:CTEvsBMRDFS} highlights that \texttt{BFDN} is the only algorithm to outperform \texttt{CTE} of \cite{fraigniaud2006collective} in an unbounded range of parameters $(n,D)$. Indeed, the other competitor, \texttt{Yo}*, is dominated by \texttt{CTE} when $n\geq e^k$ or when $D\geq e^{\log(k)^2}$. Yet, \texttt{CTE} remains the most efficient algorithm for trees with small depth, i.e. satisfying $D^2\le o_k(n)$, and its competitive ratio in $k/\log(k)$ is not uniformly surpassed by any known algorithm.

We now briefly detail the calculations that justify the picture in Figure \ref{fig:CTEvsBMRDFS}.

\paragraph{Comparison between \texttt{BFDN} and \texttt{CTE}.} Since the runtime of any collaborative tree algorithm exceeds $n/k$ and $D$, it is sufficient to compare the suboptimal terms of both algorithms which are  $D^2\log(k)$ and $n/\log(k)$ for \texttt{BFDN} and \texttt{CTE} respectively. It therefore turns out that \texttt{BFDN} is faster than $\texttt{CTE}$ in the range $D^2\log(k)^2\leq n$.  

\paragraph{Comparison between \texttt{CTE} and \texttt{Yo}*.} First, we simplified the runtime of \texttt{Yo}* to $\Ocal(\log(n)n/k+D)$, which gives that it can outperform the $\Ocal(n/\log(k)+D)$ of \cite{fraigniaud2006collective} only in the range $n\leq e^{k/\log(k)}$ which we extend to $n\leq e^{k}$ in the picture. After, we simplified the runtime of \texttt{Yo}* to $\Ocal(e^{\sqrt{\log(D)}}n/k+D)$ to obtain the range $D\leq e^{\log(k)^2}$. Finally, we simplified the runtime of \texttt{Yo}* to $D\log(n)\log(k)$ to get that \texttt{CTE} outperforms \texttt{Yo}* for  
trees satisfying $D\geq \frac{n}{\log(n)}\log(k)^2$.

\paragraph{Comparison between \texttt{BFDN} and \texttt{Yo}*.} We used the comparisons above for $e^k\leq n$ or $e^{\log(k)^2}\leq D$, and completed by the following simplification of the runtime of \texttt{Yo}* to $\Ocal(\log(k)n/k+D)$. \texttt{BFDN} is thus faster than \texttt{Yo}* when $\log(k)D^2\leq\log(k)n/k$, that is when $kD^2\leq n/k$.

\paragraph{Comparison between $\texttt{BFDN}_\ell$ and $\texttt{CTE}$.} We note that $\texttt{BFDN}_\ell$ may outperform $\texttt{CTE}$ only if $k^{1/\ell}>\log(k)$, or equivalently if $\ell <\frac{\log(k)}{\log(\log(k))}$, which we assumed in the caption of the Figure. Under this condition, $\texttt{BFDN}_\ell$ outperforms $\texttt{CTE}$ if $2^\ell \log(k)D^{1+1/\ell}<\frac{n}{\log(k)}$. Since we have $2^\ell <k$, this condition is met if $D<\frac{1}{k\log(k)^2}n^{\ell/(\ell+1)}$.

\paragraph{Comparison between $\texttt{BFDN}_\ell$ and $\texttt{BFDN}$.} If $n/k> D^2$, if is clear that $\texttt{BFDN}$ outperforms $\texttt{BFDN}_\ell$. On the other hand, if $n/k^{1/\ell}< D^2$, $\texttt{BFDN}_\ell$ outperforms $\texttt{BFDN}$.
\newpage

\section{Formal description of Anchor-based Invariants}
\label{app:anchor-based}

During the execution of an anchor-based algorithm, it is required that
the partially explored tree, the set $A\subseteq [k]$ of active robots, the anchor assignment $(v_i)_{i \in A}$, and the positions $(u_i)_{i\in [k]}$ of the robots always satisfy the following invariants:
\begin{itemize}
    \item all open nodes of the currently explored tree are in $\cup_{i\in [k]} P_T[u_i]$, \hfill (DFS Open Coverage)
    \item for any two robots $i\not= j$, all nodes in $P_T(\lca_T(u_i,u_j))$ are closed,\hfill(Parallel Positions)
    \item for all active robot $i$ such that $v_i\in P_T[u_i]$, all edges in the path from $v_i$ to $u_i$ are half explored, \hfill (Partial Exploration)
    \item  for all active robot $i\in A$, $\depth(v_i)\le d$, \hfill(Limited Anchor Depth)
    \item all inactive robots are located at depth at most $d$, \hfill(Inactive Depth)
    \item all open nodes of the currently explored tree are in $\cup_{i \in A} T(v_i)$, \hfill(Open Node Coverage)
    \item if $\exists i\in A$ such that either $\depth(v_i) < d$ or 
      $v_i$ is open, then  at least $k^*$  robots are active, \hfill(Shallow Activity)
    \item if all anchors $\{v_i : i\in A\}$ are at depth $d$ and are close, each active robot triggers an edge event at each round. \hfill(Deep Activity)
\end{itemize}

Initially, robots are said to be in \emph{Parallel DFS Positions} when DFS Open Coverage, Parallel Positions and Partial Exploration are all three satisfied when assuming that all robots are active and anchored at the root. One can easily check that other invariants are then also satisfied.

\paragraph{Properties of an anchor-based algorithm.}
The Open Node Coverage invariant implies that all nodes at depth less than $d'$ are closed where $d'=\min_{i\in A}\depth(v_i)$ is the minimum depth of an anchor. 
The Shallow Activity invariant implies that the number of active robots may decrease below $k^*$ only when all anchors are at depth $d$ and consequently when all nodes up to depth $d$ are closed. The Open Node Coverage  invariant also implies that for any dangling edge adjacent to a discovered node $w$, there exists at least one active robot $i$ such that $w$ is in $T(v_i)$. This implies that if all anchors are at depth $d$ and if $i$ is the last robot with anchor $v_i$, it cannot become inactive unless $T(v_i)$ has been completely explored. This indeed implies that the algorithm cannot terminate unless the full tree has been completely explored: as long as there remains an open node $w$, some robot $i$ must be active with an ancestor of $w$ as anchor. Recall that we require that the algorithm cannot terminate unless all robots are inactive. 

\paragraph{BFDN}
$\DFBNmath_1(k,k,d)$ is an anchor-based algorithm.
Indeed, the Open Node Coverage invariant is shown as Claim~4; the DFS Open Coverage and Partial Exploration invariants come from the similarity of \texttt{DN} moves with a DFS traversal, while the Parallel Positions invariant comes from the selection of distinct dangling edges when several robots are located at the same node.  The Limited Anchor Depth and Inactive Depth invariants are satisfied by the modification of anchor selection. The Shallow Activity invariant comes from the fact that all robots are active as long as there remain some dangling edge at depth at most $d$. Finally, the Deep Efficiency invariant comes from Claim~5 as when the algorithm runs deep, each sub-tree at depth $d+1$ which is not completely discovered contains exactly one robot performing a DFS-like traversal of the sub-tree.

We also note that we can start  $\DFBNmath_1(k,k,d)$ from any partially explored tree where robots are in Parallel DFS Positions as long as each robot $i$, which is in a position $u_i$ with open ancestors, gets anchored to a node $v_i$ of $P[u_i]$ such that all nodes of $P(v_i)$ are closed. Such a situation occurs in \texttt{BFDN} when a robot is performing \texttt{DN} moves. It is thus possible to start a robot in any such situation so that it will then behave similarly as in \texttt{BFDN}. The other robots see only closed nodes and thus get to the root according to Algorithm~\ref{alg:BFDN_algo} where they get re-anchored.

\section{Divide-depth Algorithm}\label{sec:pseudocode_div_depth}

\begin{algorithm}[H]
\caption{Divide depth algorithm $\Dcal[\Acal(k^*,k',d'); n_{team}; n_{iter}]$}\label{alg:divide-depth}
\begin{algorithmic}[1]
\Require  An anchor-based exploration algorithm $\Acal(k^*,k',d')$, integers $n_{team}\ge k^*$ and $n_{iter}\ge 1$, a partially explored tree $T$ with $k=n_{team}k'$ robots in Parallel DFS Positions and such that at most $k^*$ robots are at depth greater than $0$. 
\Ensure All nodes are discovered and closed.

\State $R\gets \set{\texttt{root}(T)}$\Comment{Set of sub-tree roots in next iteration.}
\State $A \gets \set{i\in [k] : u_i\not= \texttt{root}(T)}$\Comment{Set of robots having already progressed in $T$.}
\State All robots are active and have $\texttt{root}(T)$ as anchor.

\For{$i=1,\ldots,d/d'$}\\
  \Comment{Iteration $i$:\hfill\hfill\hfill\hfill\hfill\hfill\hfill\hfill\hfill\hfill\hfill\hfill\hfill\hfill\hfill\hfill\hfill\hfill\hfill\hfill\hfill\hfill\hfill\hfill\hfill\hfill\phantom{.}} 
  \State For all $r\in R$, let $k_r = \card{\{i \in A : v_i = r\}}$ be the number of robots having progressed in $T(r)$.
  \State Partition robots into $|R|$ teams $(B_r)_{r \in R}$ of $k'$ robots each, one per node $r\in R$: 
    \State \quad each robot $i\in A$ is assigned to $v_i$,
    \State \quad for all $r\in R$, $k'-k_r$ robots in $[k]\setminus A$ are assigned to $r$.\Comment{We rely on $k_r\le k'$ and $|R|\le n_{team}$.}
  \State All robots in team $B_r$ are assigned to anchor $r$: we set $v_i\gets r$ for all $i\in B_r\setminus A$.
  \State All robots in $\cup_{r\in R}B_r\setminus A$ are turned to active, and move to their anchor in $2(i-1)d'$ rounds.\label{line:dd-iteration-move} \Comment{Moves for rebalancing robots.}
  \State All robots in $[k]\setminus \cup_{r\in R}B_r$ are turned to inactive and wait at their current position.
  \State Each team associated to $r\in R$ initializes independently an instance $\Acal_r(k^*,k',d')$ for exploring $T(r)$.
  \State At any round, we let $A_r$ denote the set of active robots among the team exploring $T(r)$.
  \While{$\card{\cup_{r\in R}A_r}\ge k^*$} \label{line:dd-inner-while-beg}
    \State Run in parallel one round of all instances $\Acal_r(k^*,k',d')$ for $r\in R$.
  \EndWhile \label{line:dd-inner-while-beg}
  \State $A\gets \card{\cup_{r\in R}A_r}$\Comment{Overall set of active robots.}
  \State $R\gets\set{v_i : i\in A}$ \Comment{Roots of sub-trees not fully explored yet.}
\EndFor
\State Continue running instances $\Acal_r(k^*,k',d')$ of the last iteration for all $r\in R$. \Comment{Running deep.}  
\end{algorithmic}
\end{algorithm}

\begin{proof}[Proof of Proposition~\ref{prop:divide-depth}]
  We first check that all invariants are preserved by induction on the iteration number $i$. The main argument is that all anchors are at depth $i\cdot d'$ after Iteration~$i$.
  We require that the DFS Open Coverage, Parallel DFS Positions and Partial Exploration invariants are satisfied by the initial positions of robots. All remaining invariants are also satisfied as the only initial anchor is at depth zero. 
  Assume that all invariants are satisfied up to the beginning of Iteration~$i$, and that nodes in $R$ are at depth $(i-1)d'$.
  
  The Inactive Depth invariant ensures that inactive robots at the end of the previous iteration are at depth $(i-1)d'$ or less, and moving them according to Line~\ref{line:dd-iteration-move} can indeed be done within $2(i-1)d'$ rounds. Moreover, the Open Node Coverage invariant ensures that all nodes at depth less than $(i-1)d'$ are closed, and these movements preserve the DFS Open Coverage and Parallel  Positions invariants. The Partial Exploration invariant is also preserved since these robots are not located in the sub-tree of their anchor. These $(i-1)d'$ rounds also preserve Anchor Depth and Open Node Coverage invariants as the anchors $R$ of nodes active in the last round of the previous iteration remain their anchor, while other nodes are assigned to one of the anchors in $R$.

  The fact that robots are initially in Parallel DFS Positions in each instance $\Acal_r(k^*,k',d')$ for $r\in R$ comes from the preservation of the DFS Open Coverage, Parallel Positions, and Partial Exploration invariants at the end of the previous round as the root $r$ was the anchor of robots that are not located at $r$.
  Now, as all instances $\Acal_r(k^*,k',d')$ for $r\in R$ run in disjoint sub-trees, the DFS Open Coverage, Parallel Positions, Partial Exploration, Anchor Depth and Open Node Coverage invariants are also preserved during the rest of the iteration since each $\Acal_r(k^*,k',d')$ is anchor-based.
  Similarly, the Inactive Depth invariant is satisfied as its variant in instances $\Acal_r(k^*,k',d')$ imply that inactive nodes are at depth $(i-1)d'+d'=i\cdot d'\le d$ at most. The Shallow Activity invariant is preserved as long as at least one instance $\Acal_r(k^*,k',d')$ is not running deep according to the Shallow Activity invariant for that instance. This means that the number of overall active robots can drop below $k^*$ only when all instances are running deep, implying that all anchors are then at depth $(i-1)d'+d'=i\cdot d'$. Note that the Open Node Coverage invariant then implies that all open nodes are in the sub-trees rooted at the anchors of the robots that were active in the last round. The exploration can thus be reduced to these at most $k^*$ sub-trees as claimed in the description of the divide depth functor.

Finally, the algorithm starts running deep only when all anchors are at depth $d$ and are all closed. This can happen only towards the end of the last iteration when all instances are running deep. The reason is that if an instance is not running deep, it has at least $k^*$ active robots by the Shallow Activity invariant and the termination condition of the inner while loop at Line~\ref{line:dd-inner-while-beg} is not met.   The Deep Activity invariant then follows from the fact that instances are running in pairwise disjoint sub-trees and all satisfy the Deep Activity invariant.

  This achieves the proof that $\Dcal[\Acal(k^*,k',d'); n_{team}; n_{iter}]$ is correct and that it is an anchor-based exploration algorithm.

\smallskip

The proof for $f$-shallow efficiency is given in Section~\ref{sec:trading}.
\end{proof}

\end{document}